\documentclass[journal]{IEEEtran}
\usepackage[utf8]{inputenc}
\usepackage[T1]{fontenc}

\usepackage{graphicx}
\usepackage{subfig}
\usepackage{float}

\usepackage{nth}
\usepackage[noadjust]{cite}
\usepackage[hidelinks]{hyperref}
\usepackage[shortlabels, inline]{enumitem}

\usepackage{array, makecell, multirow, booktabs}
\setcellgapes{2.8pt}

\usepackage{amsmath, amsthm, amssymb, amsfonts, mathtools, optidef}

\theoremstyle{theorem}
\newtheorem{theorem}{Theorem}

\makeatletter
\def\th@definition{
\thm@headfont{\itshape} 
\thm@notefont{} 
{}
}
\makeatother
\theoremstyle{definition}
\newtheorem{definition}{Definition}

\usepackage[linesnumbered, algoruled, boxed, lined]{algorithm2e}
\SetKwRepeat{Do}{do}{while} 

\newcommand{\sol}{\textsc{Whistle}\xspace}
\newcommand{\esol}{\textsc{extended~Whistle}\xspace}
\newcommand{\solution}{\textbf{Whi}ch \textbf{S}ervice \textbf{T}o off\textbf{l}oad for remote \textbf{E}xecution\xspace}
\newcommand{\eg}{\textit{e.g.,~}}
\newcommand{\ie}{\textit{i.e.,~}}
\newcommand{\etal}{\textit{et al.}~}

\usepackage{color, xcolor}
\newcommand{\boubakr}{\color{black}}

\begin{document}
\bstctlcite{IEEEexample:BSTcontrol}

\title{Matching-based Service Offloading for Compute-less Driven IoT Networks}

\author{
	\IEEEauthorblockN{
		Boubakr Nour, \IEEEmembership{Member, IEEE}, and
		Soumaya Cherkaoui, \IEEEmembership{Senior Member, IEEE}
	}
	
	\thanks{
		B. Nour and S. Cherkaoui are with are with the INTERLAB Research Laboratory, Faculty of Engineering, Department of Electrical and Computer Science Engineering, Université de Sherbrooke, Sherbrooke (QC) J1K 2R1, Canada 
		(e-mails: boubakr.nour@usherbrooke.ca, soumaya.cherkaoui@usherbrooke.ca).}
}

\maketitle

\begin{abstract}
	With the advent of the Internet of Things (IoT) and 5G networks, edge computing is offering new opportunities for business model and use cases innovations. Service providers can now virtualize the cloud beyond the data center to meet the latency, data sovereignty, reliability, and interoperability requirements. Yet, many new applications (\eg augmented reality, virtual reality, artificial intelligence) are computation-intensive and delay-sensitivity. These applications are invoked heavily with similar inputs that could lead to the same output. Compute-less networks aim to implement a network with a minimum amount of computation and communication. This can be realized by offloading prevalent services to the edge and thus minimizing communication in the core network and eliminating redundant computations using the computation reuse concept. In this paper, we present matching-based services offloading schemes for compute-less IoT networks. We adopt the matching theory to match service offloading to the appropriate edge server(s). Specifically, we design, \sol, a vertical many-to-many offloading scheme that aims to offload the most invoked and highly reusable services to the appropriate edge servers. We further extend \sol to provide horizontal one-to-many computation reuse sharing among edge servers which leads to bouncing less computation back to the cloud. We evaluate the efficiency and effectiveness of \sol with a real-world dataset. The obtained findings show that \sol is able to accelerate the tasks completion time by 20\%, reduce the computation up to 77\%, and decrease the communication up to 71\%. Theoretical analyses also prove the stability of the designed schemes.
\end{abstract}


\IEEEpeerreviewmaketitle

\section{Introduction}
\label{sec:introduction}
The raise number of connected Internet of Things (IoT) devices and the massive amount of generated data led to radical changes not only in the nature of applications' design but also in their requirements and needs~\cite{samuel2019making}. For instance, real-time applications demand response times in the sub of $100~ms$ or even the sub of $10~ms$ range~\cite{alnoman2019emerging}. With the massive volume of generated low-value time-series data and the complex data analysis process, transmitting these data to distant cloud servers might be costly over a shared core network and shatter the low/ultra-low latency needs, which in return fail to meet the desired quality of service (QoS) requirements~\cite{xiong2018extend}.

Edge computing~\cite{filali2020multi} has been introduced to overcome numerous issues, such as low/ultra-low latency, lack of high-speed connectivity, security, data sovereignty, reliability, and interoperability with legacy systems. Edge computing enables service providers to virtualize the cloud beyond their data center~\cite{laroui2021edge}. Indeed, computations and workloads created in the cloud can now be offloaded towards edge servers that are located near to users (essentially at the edge network, one hop from the application). Appropriate data analysis and computation can be performed at the edge network and hence meet the desired QoS requirements~\cite{ni2019toward}.
Cloud and edge computing are complementary in relation to analytic processing workloads and computation. Indeed, the edge and cloud servers need to maintain a communication channel to exchange statistics regarding the change in the workload, computation, and requests' distribution. These statistics will lead to an efficient offloading decision, where only worthiness services will be moved to edge servers~\cite{wei2018dynamic}.

With a deep view of these statistics (\eg for both communication and computation), we found that the same service can be invoked multiple times with identical, similar, or semi-similar input data. These data, even with the non-identical feature, lead to the same output~\cite{sanadhya2012asymmetric}. For example, object/person detection and recognition in a smart city, road signs detection and identification in smart vehicles, crossing people and obstacles detection in autonomous vehicles, handwriting detection in optical character recognition share the same characteristics (regardless of the application's context) where the same service will most likely receive multiple invocations with different input data that produce the same output (results)~\cite{he2016exploiting}. Conventional computation at the edge network blindly computes redundant tasks multiple times (\eg vehicles driving on the same road will redundantly detect the same road signs), which deficiently utilizes the edge resources and impacts the QoS.

\begin{figure*}[!t]
	\centering
	\includegraphics[width=\linewidth]{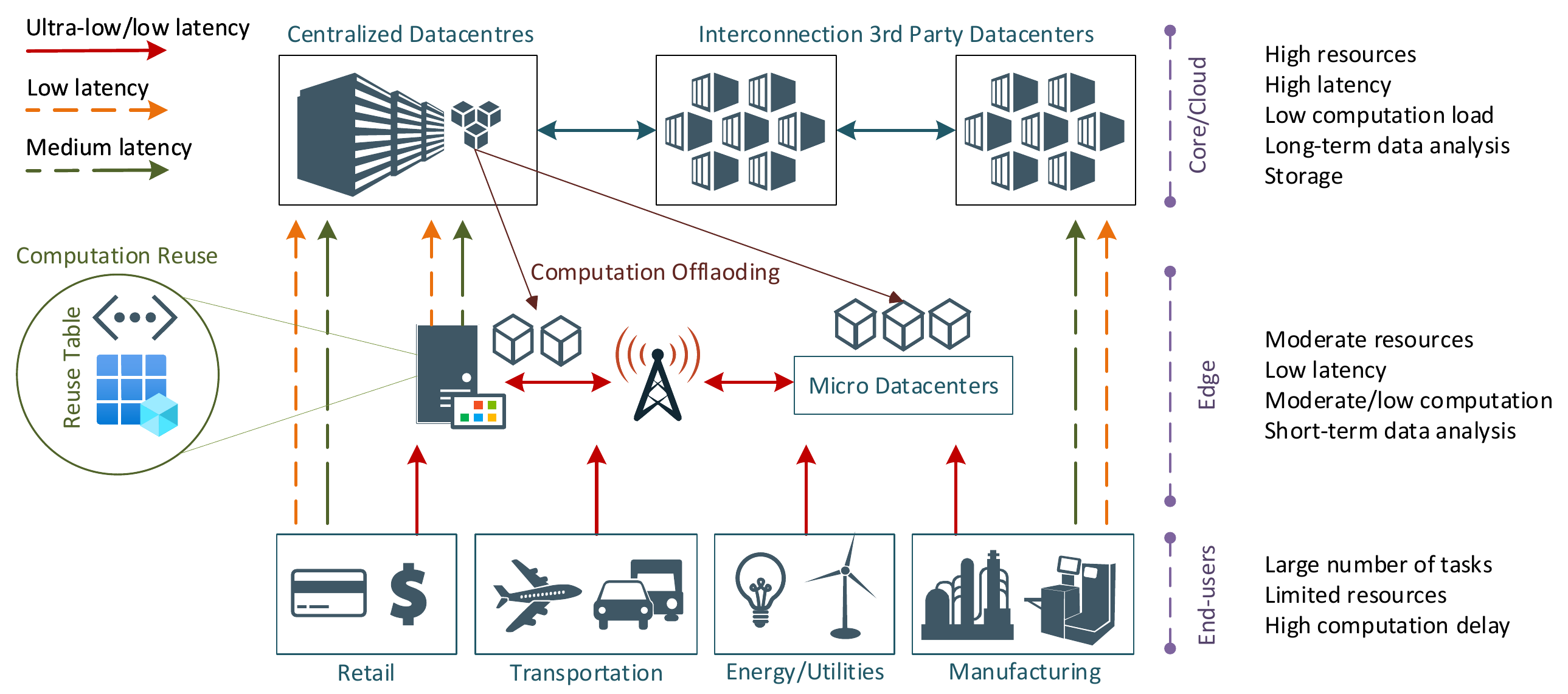}
	\caption{Reference network for computation offloading at the edge.}
	\label{fig:reference_network}
\end{figure*}

Compute-less networking~\cite{nour2020computeless}, an extended design on top of edge computing, tends to further optimize edge computing. It aims to design a network with a minimum of computation and communication costs~\cite{nour2021howfar}. In doing so, different communication and computation techniques need to take place, including computation offloading, in-network computing, and computation reuse~\cite{tang2020communication}. The first one attempts to offload the most used services from distant cloud servers to edge network~\cite{wang2019edge}, the second endeavors to provide computation at the network level~\cite{albalawi2019inca}, while the last one serves to store the previously computed tasks on a local data structure, known as Reuse Table, and use their results to satisfy the newly received similar tasks~\cite{nour2021whispering}. 
Indeed, by leveraging these mechanisms, a compute-less network will perform less computation and communication (\eg fewer data traversing the network, fast computation at the edge), which will improve the quality of service/experience, optimize the resource utilization, and enhance the overall network performance~\cite{nour2021federated, krol2019compute}.

Existing computation offloading schemes~\cite{mach2017mobile} aims to offload services from the cloud to the edge by targeting a specific objective (\eg delay, resources, budget). Based on the frequency of invocation and users' location, the same service can be offloaded to multiple edge servers. 
{\boubakr Various optimization techniques have been used to solve the computation offloading problems, such as matching theory, genetic algorithms, simulated annealing, greedy algorithms~\cite{lin2019computation}. Matching game, in particular, comes up with an efficient matching strategy to match multiple services with multiple servers to satisfy some constraints~\cite{abouaomar2018matching}.}
In this work, we design a compute-less driven IoT network by adopting matching theory. {\boubakr Compared with conventional computation offloading solutions, we further append the computation reuse feature to implement a QoS-aware offloading scheme. The combination of these techniques, at the offloading decision level, helps in improving the decision-making and the quality of service, which consequently enhances the overall network performance. Compared with other offloading schemes, the matching theory provides dynamic and scalable decision-making considering the nature of services, networks, and users' demands.}
The major contributions of this work are:

\begin{itemize}
	\item We investigate service offloading in compute-less networks. Instead of focusing only on the computation, we join the computation caching and reuse in the offloading formula.
	
	\item We formulate the service offloading problem as a matching theory problem. We design, \sol, a vertical service offloading scheme. The verticalness refers to the relationship between cloud and edge servers (cloud $\updownarrow$ edge) in deciding/providing computation. \sol is a many-to-many matching game scheme that tends to offload the most invoked services with high reusability gain from the cloud to a set of edge servers.
	
	\item We design, \esol, a horizontal computation offloading scheme between edge servers. The horizontalness refers to the cooperative relationship between edge servers (edge $\leftrightarrow$ edge) in providing costless computation. \esol is a one-to-one matching scheme that aims to satisfy the tasks for the non-offloaded services by forwarding them to the near edge server rather than bouncing them back to a distant cloud server.
	
	\item We provide an extensive simulation using Alibaba Cluster Dataset. The obtained results show the effectiveness and efficiency of the designed schemes in reducing both communication and computation as well as meeting the desired QoS requirements.
\end{itemize}

The remainder of the paper is organized as follows.
Section~\ref{sec:model} introduces the system model, network entities, as well as problem formulation.
Section~\ref{sec:offloading_scheme} presents the vertical service offloading scheme (\sol) between cloud and edge servers.
Similarly, Section~\ref{sec:sharing_scheme} presents the extended computation offloading scheme (\esol) that tends to share the computation reuse between multi-edge servers.
Section~\ref{sec:evaluation} presents the evaluation performance and discusses the obtained results.
In Section~\ref{sec:related_work}, we overview existing work that falls under the umbrella of compute-less networks and computation offloading.
Finally, we conclude the paper in Section~\ref{sec:conclusion}.

\section{System Model \& Problem Formulation}
\label{sec:model}

\subsection{Network Architecture}
In this work, we consider a cloud-edge network, as illustrated in Fig.~\ref{fig:reference_network}, where a service provider manages a distant cloud server $C$ and a set of edge servers distributed geographically $E, |E| = n$.
The cloud server, with immense computation resources, manages a set of services $S, |S| = m$. The service provider has the authority to offload its services to (multiple) edge servers based on the available resources and the frequency of invocations.
For readability purposes, the index on edge servers ($E$) is denoted by $e$ and the index on services ($S$) is denoted by $s$. In the rest of this model, notations are presented for a period of time, $\tau$, and are summarized in Table~\ref{tab:notation}.

\begin{table}[!t]
	\centering
	\makegapedcells
	\caption{Summary of the most used notation in this paper.}
	\label{tab:notation}
	\begin{tabular}{r l}
		\toprule
		\textit{Notation} &
		\textit{Description} \\
		\midrule
		$S$		& Set of services \\
		$E$		& Set of edge servers \\
		$f^c$ 	& Resources capacities at the cloud \\
		$b^c$ 	& Minimum bandwidth towards the cloud \\
		$f^e$ 	& Resources capacities at the edge $e$ \\
		$b^e$ 	& Minimum bandwidth towards the edge $e$ \\
		$S^e$ 	& Set of services offloaded to edge $e$ \\
		$T^s$	& List of tasks invoking the service $s$ \\
		$I_t$	& Task's input data \\
		$F_t$ 	& Task's execution complexity \\
		$D_t$ 	& Task's output data \\
		$\Gamma(t)$ & Task's communication cost \\
		$\chi(t)$   & Task's execution cost \\
		$\zeta_t$   & Task's completion cost \\
		$L$ 	    & Lookup cost in reuse table \\
		$\upsilon_e^s$ & Network delay of offloading service $s$ to edge $e$ \\
		$x_t^e$     & 1-0 Offloading variable:\\
		& $x_t^e = 1$: task computation is offloaded to edge $e$ \\
		$y_t$       & 1-0 Reuse variable:\\
		& $y_t = 1$: computation reuse is applied to execute task $t$ \\
		$r_t$       & 1-0 Full reuse variable:\\
		& $r_t = 1$: Task $t$ is executed via full reuse \\
		\bottomrule
	\end{tabular}
\end{table}

\vspace{0.2cm}
\textbf{Service Modeling.}
Services are initially offered by the cloud server $C$ and could be offloaded to different edge servers $\{e\} \in E$. A service $s \in S$ is defined by the following tuple:

$$ s = \langle~\delta_s, \sigma_s, \rho_s~\rangle, $$

\noindent where $\delta_s$ denotes the service granularity, $\sigma_s$ denotes the service potential reusability, and $\rho_s$ denotes the service punishment metric.

\begin{definition}[Service Granularity]
	It defines how much the received input data is similar. It is calculated, as shown in Eq.~(\ref{eq:service_granularity]}), based on the total number of received data ($I_{\mathsf{r}}$) and the number of distinct data ($I_{\mathsf{d}}$), for a period of time $\tau$.	
	
	\begin{equation} \label{eq:service_granularity]}
		\delta_s = \frac{1}{1 + e^{-\frac{I_{\mathsf{r}}}{I_{\mathsf{d}}}}}
	\end{equation}
\end{definition}

\begin{definition}[Service Potential Reusability]
	It refers to the rate that a service $s$ can benefit from processing similar input data (\eg how much reuse we can achieve). This metric, as shown in Eq.~(\ref{eq:service_potentioal_reusability}), is based on the service granularity.
	
	\begin{equation} \label{eq:service_potentioal_reusability}
		\sigma_s = \frac{1}{1 + e^{-\frac{1}{\delta_s}}}
	\end{equation}
\end{definition}

\begin{definition}[Service Punishment]
	It aims to diminish the chances of an already evicted service from being re-offloaded to the same edge server. This eviction mainly corresponds to a change in the received task's distribution of the service. Thus, re-offloading the same service may not help in achieving better performance. It is calculated, as shown in Eq.~(\ref{eq:service_punishement}), based on the service potential reusability.
	
	\begin{equation} \label{eq:service_punishement}
		\rho_s = 
		\begin{cases}
			& \sigma_s,    \qquad\text{if $s$ has not been evicted,} \\
			& 1 - \sigma_s,~\text{otherwise.} \\ 
		\end{cases}
	\end{equation}
\end{definition}

Let $S^e \subset S$ denote the list of services offloaded at the edge $e \in E$. A service $s \in S^e$ is defined by the following tuple:

$$ s = \langle~\delta_s, \varphi_s~\rangle, $$

\noindent where $\delta_s$ denotes the service granularity, and $\varphi_s$ denotes the gain of service offloading.

\begin{definition}[Offloading Gain]
	It indicates the reduced communication and computation that will be witnessed when the said service is offloaded. The gain is based on the average size of the received input data and the execution complexities, for a period of time $\tau$, and defined by the following tuple:
	
	$$ \varphi_s = \langle~\sigma_s.\text{avg}(\sum I_t),~~\sigma_s.\text{avg}(\sum f_t)~\rangle $$
\end{definition}

\vspace{0.2cm}
\textbf{Task Modeling.}
Due to their limitations in resources, IoT users send their tasks for remote execution, which we consider an $M/M/1$ queuing system following Little's Law.
Let $T^s$ denote the list of tasks invoking the service $s \in S$. A task $t \in T^s$ is defined with the following tuple:

$$ t = \langle~I_t, F_t, D_t~\rangle, $$

\noindent where
$I_t$ is the task's input data to execute the target service,
$F_t$ is the task's execution complexity and defines the required resources to execute the given task with the provided input data, and
$D_t$ is output data after a successful execution.

\vspace{0.2cm}
\textbf{Communication Modeling.}
Let $\Gamma(t)$ denote the task's communication cost. The cost is based on the task's input data ($I_t$) and the minimum bandwidth between end-user and cloud ($b^c$) or edge server ($b^e$). The communication cost is defined as shown in Eq.~(\ref{eq:communication_cost}).

\begin{equation} \label{eq:communication_cost}
	\Gamma(t) = x_t^e.\frac{I_t}{b^e} + (1-x_t^e).\frac{I_t}{b^c}, \quad \forall e \in E
\end{equation}

\noindent where $x_t^e$ denotes whether a task $t$ is sent to an edge server ($x_t^e = 1$) or the cloud ($x_t^e = 0$) for execution. We call $x_t^e$ the offloading decision variable.

\vspace{0.2cm}
\textbf{Computation Modeling.}
Let $\chi(t)$ denote the task's execution cost. $\chi(t)$ depends on the task's complexity ($F_t$) and the server's resources. Let $f^c$ and $f^e$ denote the resources at the cloud and the edge server, respectively. The execution cost is defined as shown in Eq.~(\ref{eq:computation_cost}).

\begin{equation}
	\label{eq:computation_cost}
	\chi(t) = x_t^e.\frac{F_t}{f^e} + (1-x_t^e).\frac{F_t}{f^c}, \quad \forall e \in E
\end{equation}

If a service $s$ is offloaded to an edge, its associated tasks will be executed at the assigned edge and will not be bounced back to the cloud server.

If computation reuse is applied, we differentiate two cases:
\begin{enumerate*}[(i)]
	\item {\em full computation reuse}: the stored output in the Reuse Table will be used to fully satisfy the received task, the execution cost is equal to the lookup cost to find a match in the Reuse Table; and
	\item {\em partial computation reuse}: the stored output in the Reuse Table will be used to partially satisfy the received task while the rest of the task will be executed, the execution cost is equal to the sum of the lookup process and reset of computation to fulfill the whole task's execution.
\end{enumerate*}
The cost is then defined as shown in Eq.~(\ref{eq:cost_reuse}).

\begin{equation} \label{eq:cost_reuse}
	\eta(t) = r_t.L + (1-r_t).(L + \chi(t'))		
\end{equation}

\noindent where $r_t$ denotes whether a task $t$ is satisfied via a full computation reuse ($r_t = 1$) or a partial computation reuse ($r_t = 0$), and $t' \subset t$ is the remained part of the task $t$ to be executed.

Finally, the overall task's completion cost ($\zeta_t$) is defined as shown in Eq.~(\ref{eq:cost_total}).

\begin{multline}
	\label{eq:cost_total}
	\zeta_t =
	\underbrace{(1 - x_t^e).(\frac{I_t}{b^c} + \frac{F_t}{f^c})}_\text{At cloud} +
	\underbrace{x_t^e(1 - y_t).(\frac{I_t}{b^e} + \frac{F_t}{f^e})}_\text{At edge without reuse} + \\
	\underbrace{x_t^e y_t.(\frac{I_t}{b^e} + r_t.L)}_\text{At edge with full reuse} + 
	\underbrace{x_t^e y_t.(\frac{I_t}{b^e} + (1-r_t).(L + \chi(t')))}_\text{At edge with partial reuse}
	\\ = \Gamma(t) + (1 - y_t).\chi(t) + y_t.\eta(t)
\end{multline}

\noindent where $y_t$ denotes if computation reuse ($y_t = 1$) or computation from scratch ($y_t = 0$) is applied at the edge server. 

\subsection{Problem Formulation}
For any time period, $\tau$, the service provider determines which services need to be offloaded to the which edge server for remote execution in order to minimize the overall task's completion cost. This results in maximizing the computation reuse gain and reducing the overall edge/network resource utilization based on an optimal service offloading and computation reuse strategy. Thus, we can formulate the service offloading problem as:

\begin{equation} \label{eq:obj}
	\underset{x, y, r}{\min}~~\sum_{t = 1}^{n} \sum_{e = 1}^{m} \zeta_t
\end{equation}

subject to,

\begin{subequations}
	\begin{equation} \label{eq:cons1}
		\sum_{t \in T^s} x_t^e F_t \leq f^e \qquad \forall e \in E,
	\end{equation}
	\begin{equation} \label{eq:cons2}
		\sum_{t \in T^s} x_t^e I_t \leq b^e \qquad \forall e \in E,
	\end{equation}
	\begin{equation} \label{eq:cons3}
		x_t^e \in \{0, 1\} \quad \forall t \in T, \forall e \in E,
	\end{equation}
	\begin{equation} \label{eq:cons4}
		y_t, r_t \in \{0, 1\} \qquad \forall t \in T,
	\end{equation}
	\begin{equation} \label{eq:cons5}
		y_t \geq x_t^e \quad \forall t \in T, \forall e \in E,
	\end{equation}
	\begin{equation} \label{eq:cons6}
		r_t \geq y_t^e \quad \forall t \in T.
	\end{equation}
\end{subequations}

The objective function (\ref{eq:obj}) tends to minimize the overall task completion cost, subject to a set of constraints: 
(\ref{eq:cons1}) the aggregate task complexity for a given edge server should not exceed the edge's resource capacity, 
(\ref{eq:cons2}) the overall data traversed the network should not exceed the available bandwidth,
(\ref{eq:cons3})-(\ref{eq:cons4}) the decision variables are binary, and
(\ref{eq:cons5})-(\ref{eq:cons6}) impose the variables' boundary.

\vspace{0.2cm}
\begin{theorem}
	The problem (\ref{eq:obj}) is NP-hard.
\end{theorem}

\begin{proof}[Proof]
	The objective of cloud-edge service offloading is to minimize the overall execution time. The problem formulated in (\ref{eq:obj}) can be represented as a 0-1 Knapsack problem. The list of offload-able services is considered as items, offloading gain ($\varphi$) as item's weight, and the capacity of the edge server as the maximum weight capacity. The 0-1 Knapsack problem is an NP-complete, and therefore, the presented problem in (\ref{eq:obj}) is NP-hard.
\end{proof}

\section{\sol: \solution}
\label{sec:offloading_scheme}
As its name stands for, \sol aims at answering the question of which services need to be offloaded for remote execution? To achieve this goal, \sol relies on matching theory. The fundamental objective of the matching theory is to find an optimal strategy to match players from two separate sets by formulating relationships that benefit all interested players while taking into account their particular preferences. In this work, we consider a vertical offloading scenario with two sets: edge servers $E$ and services $S$ hosted at the cloud. The objective is to match the most used services (in terms of invocation frequency and high computation reuse) with the appropriate edge servers. Communication, computation, and reuse have been taking into consideration when defining the preferences for each player.
In this section, we will start by providing details about the matching theory and how it relates to our offloading problem, and then we will dive into technical details of the proposed offloading schemes.

\subsection{Definitions}
In the following, we recall the basic definitions of matching game theory and relate them to our service offloading problem.

\begin{definition}
	Given two sets of players $S$ and $E$, a matching game is a pair $(S, E)$ on which two preference relations $\succ_s$ and $\succ_e$ are defined. These relations allow each player to express its preferences over the opposite players, \eg which service $s \in S$ prefers to be offloaded on which edge server $e \in E$, and vice versa.
\end{definition}

A matching game produces a matching function $\mu$, which is defined as follows.

\begin{definition}
	Given two sets of players $S$ and $E$, the function $\mu: S \times E \rightarrow S \times E$ is a matching function if conditions (\ref{eq:matching_fun_1}) and (\ref{eq:matching_fun_2}) are verified.
	
	\begin{subequations}
		\begin{align}
			\mu(\mu(w)) = w,~~\forall w \in S \times E, \label{eq:matching_fun_1}\\
			\mu(s) \in E, \mu(e) \in S,~~\forall s \in S, \forall e \in E. \label{eq:matching_fun_2}
		\end{align}
	\end{subequations}
	
\end{definition}

The matching function $\mu(\cdot)$ exploits the preference relations $\succ$ to determine the relation for a given player on one set to a given player on the opposite set. For each element on one set, the preference relation reflects the level of satisfaction in being matched with each element on the other set, and vice versa.

\begin{definition}
	A utility function is utilized to quantify the preferences between elements belonging to two different sets.
	Let $\Theta_s(\cdot)$ and $\Phi_e(\cdot)$ denote the utility functions of service $s$ and edge server $e$, respectively. If $\Theta_s(e_1) > \Theta_s(e_2)$, then the service $s$ prefers the edge server $e_1$ to the edge server $e_2$ on where it can be offloaded. Such situation is expressed by the preference relation $e_1 \succ_s e_2$. Similarly, $s_1 \succ_e s_2$ means that the edge server $e$ prefers the service $s_1$ to $s_2$ to be offloaded on it, hence the associated preference relation is $\Phi_e(s_1) > \Phi_e(s_2)$.
\end{definition}

\subsection{\sol: Service Offloading Scheme}
A vertical computation offloading aims at moving services from the cloud to edge servers. Fig.~\ref{fig:offloading_scheme} depicts the overall principles of the proposed offloading scheme, namely \sol. \sol tends to offload multiple services to multiple edge servers, which can be formed as a many-to-many matching. Hence, it builds a communication channel between the cloud server that hosts services and a set of edge servers ready to host the offloaded services. The channel helps in exchanging statistical information that is used for offloading decision-making.
In order to minimize the task's completion time, \sol requires a matching between services and available edge servers. Thus, generating preference lists is necessary before proceeding with optimal offloading decisions.

\begin{figure}[!b]
	\centering
	\includegraphics[width=.9\linewidth]{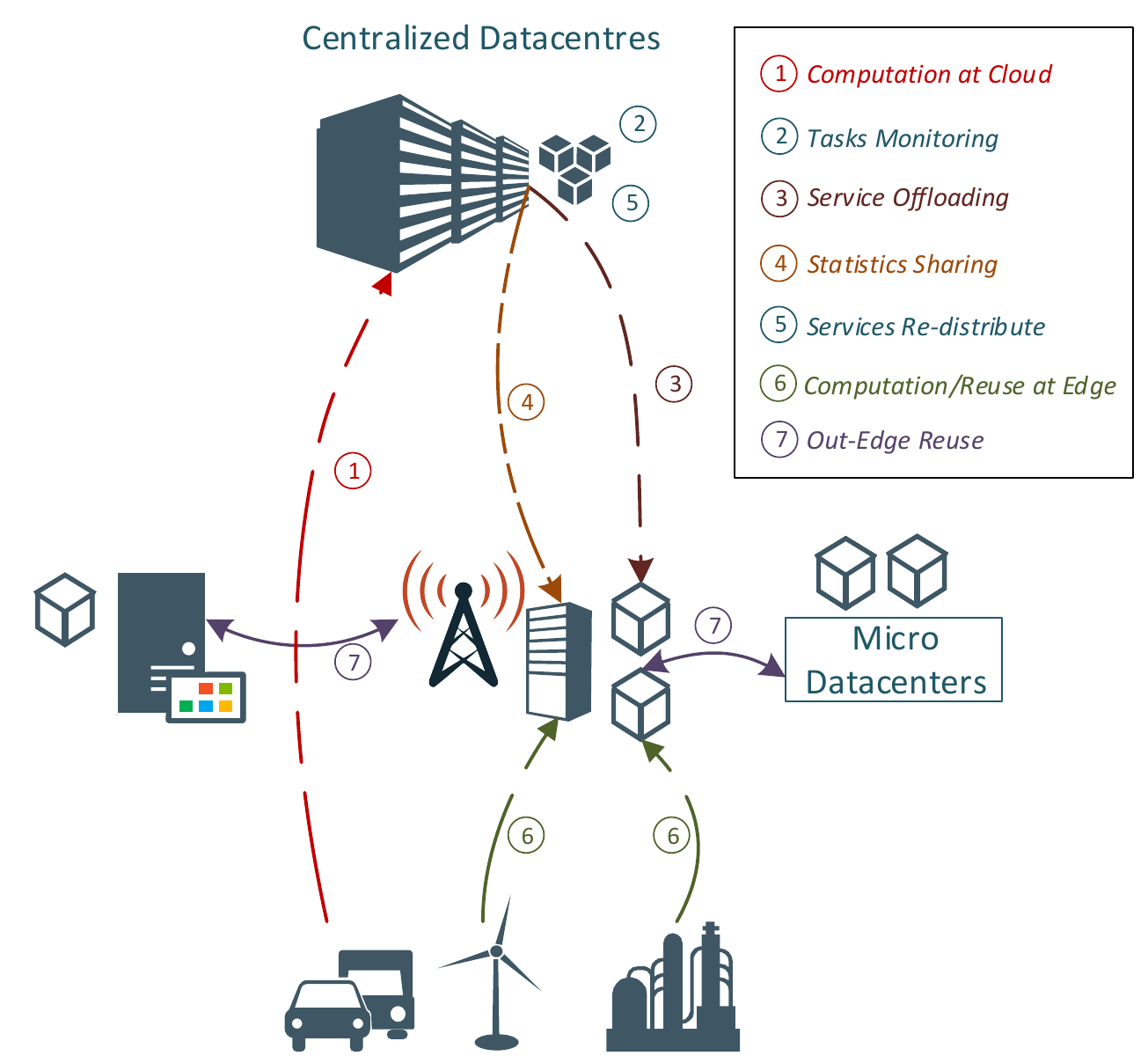}
	\caption{Overview of \sol service offloading scheme.}
	\label{fig:offloading_scheme}
\end{figure}

\vspace{0.2cm}
\subsubsection{Preference Lists}
Each service at the cloud level seeks to be offloaded to at least one edge server to minimize the communication and computation costs. Similarly, each edge server attempts to host as many services to satisfy the large number of received tasks to execute. Each player manages to satisfy some quality of service requirements related to communication, computation, and resources.

\vspace{0.2cm}
\textit{Service's Preference List.}
For a given service $s \in S$, the utility function $\Theta_s(\cdot)$ is evaluated by considering the communication cost, the service potential reusability, and the service punishment, and is given as shown in Eq.~(\ref{eq:service_preference_list}).

\begin{equation} \label{eq:service_preference_list} 
	\Theta_s(e) = \frac{1}{\Gamma(s) + \sigma_s - \rho_s}, \quad \forall e \in E.
\end{equation}

\vspace{0.2cm}
\textit{Edge's Preference List.}
For a given edge server $e \in E$, the utility function $\Phi_e(\cdot)$ is evaluated by considering the computation cost, the reuse cost, and the offloading gain, and is given as shown in Eq.~(\ref{eq:edge_preference_list}).

\begin{equation} \label{eq:edge_preference_list}
	\Phi_e(s) = \frac{1}{\chi(s) + \eta(s) + \varphi_s}, \quad \forall s \in S.
\end{equation}

\vspace{0.2cm}
\subsubsection{\sol Scheme}
In order to ensure optimal offloading results, \sol tends to meet the preference lists of both services and edge servers with the best possible match. The offloading is achieved by selecting the convenient services to the adequate edge server(s), while the edge server hosts the relevant service(s). In doing so, \sol is performed in two main phases and two-matching stage, as illustrated in Algorithm~\ref{alg:offlaoding_alg}, and detailed as follows:

\begin{algorithm}[!t]
	\For{(each period of time $\tau$)}{
		Monitor received tasks\;
		Calculate service granularity ($\delta_s$)\;
		Calculate service potential reusability ($\sigma_s$);\
	}
	
	\For{(each service $s \in S$)}{
		Construct the preference list $\Theta_s(\cdot)$ (Eq.~\ref{eq:service_preference_list})\;
	}
	
	\For{(each edge server $e \in E$)}{
		Construct the preference list $\Phi_e(\cdot)$ (Eq.~\ref{eq:edge_preference_list})\;
	}
	
	\For{(each requested edge server $s$)}{
		\uIf{(capacity is not full)}{
			Offload the proposed service to edge server $e$\;
		}\Else{
			Reject the proposed service with \textit{DAA}\;
		}
	}
	
	\Do{(all services are not offloaded)}{
		\For{(each non-offloaded service $s$)}{
			Update $\Theta_s(\cdot)$ (Eq.~\ref{eq:service_preference_list})\;
			Propose the prefered edge server\;
		}
		
		\For{(each requested edge server $e$)}{
			Accept the preferred service among candidates\;
			Reject other services\;
		}
	}
	
	\caption{\sol offloading scheme.}
	\label{alg:offlaoding_alg}
\end{algorithm}

\vspace{0.2cm}
\textbf{\textit{Phase 1.} Jump-Starting Phase:}
This phase is executed before any offloading decision (Algorithm~\ref{alg:offlaoding_alg}, Lines 1-5), where all tasks are forwarded directly to the cloud server for remote execution. At this phase, the cloud server has no prior information about the distribution of tasks and the utility of each service. Thus, it starts by 
\begin{enumerate*}[(i)]
	\item monitoring all the received tasks; and 
	\item calculating, for each service, the service granularity ($\delta_s$) and the service potential reusability ($\sigma_s$).
\end{enumerate*}

The end of this phase is marked by deciding what service(s) could be offloaded to which edge server(s). This decision is based on the collected and calculated information, and proceeded on two matching stages:

\vspace{0.2cm}
{\em Stage 1.}
The first matching stage begins after discovering services and edge servers that are interested in the offloading process. In this phase, both preference lists of services and edge servers are constructed according to Eqs.~\ref{eq:service_preference_list} and~\ref{eq:edge_preference_list}, respectively. The process seeks to develop context awareness of the service computation cost and reuse gain towards each conceivable edge server, taking communication costs and available resources into account (Algorithm~\ref{alg:offlaoding_alg}, Lines 6-11).

To efficiently exploit the available resources, we have extended the Deferred Acceptance Algorithm (DAA)~\cite{roth2008deferred} in order to deduce a preliminary matching (Algorithm~\ref{alg:offlaoding_alg}, Lines 12-18). Each service $s$ picks the most preferred edge server $e$ based on $\succ_s$. Similarly, as long as the edge server $e$ does not reach the maximum capacity $f^e$, it accepts any service offloading request. Supposing the maximum capacity $f^e$ is reached, and the said edge server receives a new offloading request for a service that is favored above at least one of the previously accepted requests. The edge server $e$, in this case, accepts the received service and rejects the least preferred service on its list based on $\succ_s$. Finally, if a service experiences a rejection, it does not eliminate the edge server from its preference list, but it jumps for the next preferred edge server on its list.

\vspace{0.2cm}
{\em Stage 2.}
The second stage of the \sol scheme (Algorithm~\ref{alg:offlaoding_alg}, Lines 19-28) tends to allocate services to the best-preferred edge server. For each non-offloaded service $s$, the cloud server evaluates $\Theta_s(e)$ for each $e \in E$. Similarly, each edge server $e$ evaluates $\Phi_e(s)$ for each service $s$ that is not already offloaded. Then, for each non-offloaded service $s$, the cloud server submits the offloading request to the top-ranked edge server according to $\succ_s$. When the edge server $e$ receives a service offloading request, it selects the preferred ones among those submitted. At the end of the matching, appropriate services are offloaded to the adequate edge servers, while edge servers host the relevant services.

\vspace{0.2cm}
\textbf{\textit{Phase 2.} Steady-State Phase:}
Since the frequency of service invocations may change over time (\eg change in the distribution of users, the popularity of service, mobility of users, etc.), the service provider needs to redistribute its offloaded services based on the new statistics. Thus, each edge server periodically shares the locally collected statistics with the centralized cloud server. The latter recalculates service granularity ($\delta_s$), service potential reusability ($\sigma_s$), and the service punishment ($\rho_s$), and redistributes services accordingly. Compared with the jump-starting phase, the service punishment metric is included in the procedure to measure how worth is re-offloading a service after the eviction. This will result in better preference list construction and hence better services' redistribution.

\subsection{Performance Analysis}
In the following, we analyze the performance of the proposed offloading scheme (Algorithm~\ref{alg:offlaoding_alg}).

\begin{theorem}
	The proposed algorithm achieves the stability of service offloading decision.
\end{theorem}

\begin{proof}[Proof]
	Since service invocation and preferences can vary during the service's lifetime, the study of the outcome matching stability is not a trivial dilemma. A service $s$ may prefer to be offloaded to a specific edge server $e_1 \succ_s e_2$, and abruptly it does not $e_1 \nsucc_s e_2$. The sudden and unexpected change leads to triggering an exchange request, thereby causing a potential unstable situation.
	
	According to~\cite{bodine2011peer}, a many-to-many matching game is stable if it ends up being a two-sided exchange-stable matching.
	
	\begin{definition}[Two-sided exchange-stable matching]
		a given outcome matching game $\mu$ is a two-sided exchange-stable matching if there does not exist a pair of services $(s_1, s_2)$ such that:
		
		\begin{subequations}
			\begin{equation} \label{eq:2es_1}
				\Theta_{s_1}(\mu(s_2)) \geq \Theta_{s_1}(\mu(s_1)),
			\end{equation}
			\begin{equation} \label{eq:2es_2}
				\Theta_{s_2}(\mu({s_1})) \geq \Theta_{s_2}(\mu({s_2})),
			\end{equation}
			\begin{equation} \label{eq:2es_3}
				\Phi_\mu({s_1})({s_2}) \geq \Phi_\mu({s_1})({s_1}),
			\end{equation}
			\begin{equation} \label{eq:2es_4}
				\Phi_\mu({s_2})({s_1}) \geq \Phi_\mu({s_2})({s_2}),
			\end{equation}
			\begin{equation} \label{eq:2es_5}
				\exists z \in \{s_1, s_2, \mu(s_1), \mu(s_2)\}.
			\end{equation}
		\end{subequations}
	\end{definition}
	
	\noindent so that any of the above conditions (\ref{eq:2es_1})-(\ref{eq:2es_5}) is strictly verified.
	
	According to this definition, two services can change their offloading destination only if both gain a benefit from the exchange. Similarly, if two edge servers want to exchange the offloading request of two services, the associated services have to get an advantage on that. Indeed, an exchange is approved only if both involved players (\ie services and edge servers) requesting service offloading, strictly improves their overall reuse gain and reduces their cost.
	
	Let $S_1$ denote the set of services to be offloaded at the end of the first stage of jump-starting phase. It can be clearly noted that $|S_1| > 0$ and is bounded by the capacity of the edge server. We can say that for an offloaded service $s \in S_1$, the associated edge servers cannot change their service offloading decision (\ie all services in $S_1$ are stably offloaded to the appropriate edge servers). This is because $\Theta_s(·)$ and $ \Phi_e(·)$, defined in Eqs.~(\ref{eq:service_preference_list}) and (\ref{eq:edge_preference_list}), respectively, are unable to change since all of the terms they rely on stay constant regardless of any possible exchange.
	In addition, for each service $s$ that still needs to be offloaded at the end of the first stage of jump-starting phase, we have a stable matching by providing a reductio ad absurdum.
	
	In the following, we assume that at a time instant $\tau^*$ of stage 2 completion where all services have been offloaded to the associated edge servers, there exists a services pair $(s_1, s_2)$ such that conditions (\ref{eq:2es_1})-(\ref{eq:2es_5}) are verified.
	
	Focusing on condition (\ref{eq:2es_1}), we must have:
	\begin{equation} \label{eq:theorem1-1}
		\Theta_{s_1}(e_2) \geq \Theta_{s_1}(e_1).
	\end{equation}
	
	Supposing that $s_1$ and $s_2$, at the allocation time $\tau$ with $\tau < \tau^*$, are offloaded to edge server $e_1$ and $e_2$, respectively, we must have $\Theta_{s_1}({e_1}) \geq \Theta_{s_1}(e_2)$, hence:
	
	\begin{equation} \label{eq:theorem1-2}
		\frac{1}{\tau + \upsilon_{e_1}^{s_1}} \geq \frac{1}{\tau + \upsilon_{e_2}^{s_2}}.
	\end{equation}
	where $\upsilon_e^s$ is the experienced network delay of offloading service $s$ to edge $e$.
	
	As a consequence, we must have the allocation time for $s_1$ during $\tau$ is equal to the allocation time during $\tau^*$, otherwise (\ref{eq:theorem1-1}) would be impossible. This is due to the second step of the jump-starting phase, where it does not take out any previously accepted service offloading request and by consequence, the expected waiting time on a given edge server cannot be reduced. Similarly, we can undertake the same analysis on conditions (\ref{eq:2es_2})–(\ref{eq:2es_4}). We will have at the end the allocation time for $s_1$ during $\tau$ is equal to the allocation time during $\tau^*$. As a consequence, condition (\ref{eq:2es_5}) cannot be verified. This results in contradicting our initial assumption, and hence we can conclude that the proposed algorithm is a two-sided exchange-stable matching.
\end{proof}

\section{\esol: Out-Edge Reuse Sharing Scheme}
\label{sec:sharing_scheme}
Since edge servers are limited in resources compared with boundless cloud servers, it is not feasible to offload all services from the cloud to the edge. Multiple (near) edge servers may host different services according to their preference lists. It is likely possible that an edge server receives tasks for execution in which their associated services have not been offloaded to the said server. Bouncing back these tasks to the cloud for remote execution is the only solution, which might fail to address the latency requirements. Indeed, redirecting the task's computation to a near edge server that has the requested service is better than offloading the task to the cloud.
In the following, we extend \sol to perform a horizontal computation offloading between near edge servers instead of arranging a vertical offloading with the cloud. \esol tends to match one service to multiple near edge servers, which can be formed at one-to-many matching. As long as the computation reuse is concerned, the computation reuse will be shared outside the local edge server. That being the case, the preference lists need to be adjusted accordingly considering the service reusability.

\subsection{Preference Lists}
For each received task with a non-offloaded service, the edge server tries to find the near edge server that has the requested service and is able to execute the task with the minimum cost compared to offloading the task to a distant cloud server. Similarly, each neighbor edge server tries to select which tasks can execute with the available resources.

\vspace{0.2cm}
\textit{Service's Preference List.}
For a given service $s \in S$, the utility function $\Theta_s'(\cdot)$ is evaluated by considering the communication cost and the service potential reusability, and is given as shown in Eq.~(\ref{eq:service_preference_list2}):

\begin{equation} \label{eq:service_preference_list2}
	\Theta_s'(i) = \frac{1}{\Gamma(s) + \sigma_s}, \quad \forall e \in E'.
\end{equation}
where $E'$ defines the list of neighbor edge servers located in one-hop.

\vspace{0.2cm}
\textit{Edge's Preference List.}
Similarly, for a given edge server $e \in E$, the utility function $\Phi_e'(\cdot)$ is evaluated by considering the computation cost, the reuse cost, and the offloading gain, and is given as shown in Eq.~(\ref{eq:edge_preference_list2}):

\begin{equation} \label{eq:edge_preference_list2}
	\Phi_e'(s) = \frac{1}{\chi(s) + \eta(s) + \varphi_s}, \quad \forall s \in S'. 
\end{equation}
where $S'$ defines the list of offloaded services at the said edge server.

\begin{algorithm}[!t]
	\For{(each service $s \in S'$)}{
		Construct the preference list $\Theta_s(\cdot)$ according to Eq.~\ref{eq:service_preference_list2}\;
	}
	
	\For{(each edge server $e \in E'$)}{
		Construct the preference list $\Phi_e(\cdot)$ according to Eq.~\ref{eq:edge_preference_list2}\;
	}
	
	\For{(each requested task)}{
		off $\leftarrow$ False\;
		\For{(each available edge server $e$)}{
			\uIf{capacity is not full}{
				Offload the received task to edge server $e$\;
				off $\leftarrow$ True\;
			}\Else{
				Reject the proposed edge with DAA\;
			}
		}
		\If{(off is False)}{
			Offload the task to cloud\;
		}
	}
	
	\Do{(all tasks are non-offloaded)}{
		\For{(each non-offloaded service $s$)}{
			Update $\Theta_s'(\cdot)$\;
			Propose the prefered edge server\;
		}
		
		\For{\textbf{each} (requested edge server $e$)}{
			Accept the preferred service among candidates\;
			Reject other services\;
		}
	}
	
	\caption{\esol scheme.}
	\label{alg:sharing_alg}
\end{algorithm}

\subsection{\esol Scheme}
\esol scheme (shown in Algorithm~\ref{alg:sharing_alg}) starts by constructing the preference lists for both services and neighbor edge servers (lines 1-6). An edge server after receiving a task non-associated with any local hosted service, try to find a better match with the available list of near edge servers.

The two-matching stage is conducted for all non-offloaded services at the current edge server yet available at near-edge servers. The aim is to find the best edge server that can execute the task with less computation, mainly by using computation reuse instead of computation from scratch. Hence, the current edge selects the most suitable near edge server to execute the received task based on the preference list. If a preferable edge is running out of resources, the DAA algorithm jumps to the next desirable edge server. The process is repeated until matching with an edge server that can execute the task. The best matching case is achieved with full computation reuse while the least matching case is done with computation from scratch. If no match is found, the task is forwarded to the cloud for remote execution.

\subsection{Performance Analysis}
Now, we analyze the performance of \esol scheme, a one-to-many matching strategy (Algorithm~\ref{alg:sharing_alg}).

\begin{theorem}
	The proposed algorithm achieves a stable matching.
\end{theorem}

\begin{proof}[Proof]
	Before proving the stability matching of the proposed algorithm, let us start with some definitions: 
	
	\begin{definition}[Blocking Pair]
		a service $s$ and an edge server $e$ form a blocking pair if both prefer the others than their currently matched results.
	\end{definition}
	
	\begin{definition}[Stable Matching]
		a one-to-many matching game $\mu$ is defined as stable game if it is not blocked by any pair.
	\end{definition}
	
	Assuming that the outcome of the matching game $\mu(s) = e'$, but $e$ and $s$ form a blocking pair, which means that both $e$ and $s$ prefer to be matched each other, but they have not been matched. Thus, we have $\mu(s) = e$, $e \succ_s e'$. However, according to Algorithm~\ref{alg:sharing_alg}, $\mu(s) = e$ is not the matching result, which means that the edge $e$ has abandoned the service $s$ during the matching process. Furthermore, the final match for $s$ is $e'$, that is $e' \succ_s e$. This results in contradicting our initial assumption, and hence we can conclude that the proposed computation offloading algorithm achieves a stable matching.
\end{proof}

\section{Evaluation Performance}
\label{sec:evaluation}
This section presents the evaluation performance of the proposed schemes. We will start by describing the simulation setup, evaluated metrics, and then discussing the obtained results.

\subsection{Simulation Setup}
The proposed schemes have been implemented using the Python programming language. The implementation includes computation offloading, computation reuse, and out-edge reuse sharing. We have evaluated \sol in comparison with:

\begin{itemize}
	\item \textit{Cloud}: all tasks are forwarded to the cloud for remote execution. The cloud has immense computation resources but is located far away from users.
	
	\item \textit{Native Edge}: services are offloaded to the edge server randomly. The edge server provides only computation without applying the computation reuse concept.
	
	{\boubakr \item \textit{Genetic Algorithm (GA)}: the service provider searches for a good set of services to offload aiming at minimizing the delay. The initial population is a set of services with the highest frequency, while the fitness score is calculated based on the service invocation number.}
	
	{\boubakr \item \textit{Simulated Annealing (SA)}: the service provider selects randomly a service and offloads it to the edge server. If the selected service improves the overall delay, then it is accepted.  Otherwise, the service provider makes the selection with some probability less than 1. The probability decreases exponentially with the ineffectiveness of the service selection.}
	
	\item \textit{Greedy Algorithm}: services are offloaded to the edge server using a greedy algorithm. By saying greedy, we mean that the service provider tends to fully utilize the available resources at the edge server. 
	
	\item \textit{Matching Algorithm}: service provider matches, using matching theory strategy, the existing services with available resources in order to decide the service offloading process. The decision is based on the preference lists without including the service reusability. The edge server provides only computation without utilizing the computation reuse concept.
	
\end{itemize}
{\boubakr To further demonstrate the effectiveness of computation reuse, we implemented each offloading scheme with and without reuse feature.} The evaluation has been carried on an Intel Core i7-1075 CPU system clocked at 2.6 GHz, with 16 GB of RAM, running Windows 10 Education 64-bit. Table~\ref{tab:parameters} provides a summary of the used parameters.

\begin{table}[!t]
	\centering
	\makegapedcells
	\caption{Experimental parameters.}
	\label{tab:parameters}
	\begin{tabular}{l l}
		\toprule
		\textit{Parameter} &
		\textit{Value} \\
		\midrule
		Hops to the edge   & $1$ \\
		Hops to cloud      & $[5-10]$ \\
		Edge capacity      & $6-10$ services simultaneously \\
		Number of tasks    & $[1000-10000]$ \\ 
		Tasks arrival rate & $[10-50]$ tasks per seconds \\
		Input redundancy   & $[10-80]\%$ \\
		Replacement policy & $LFU$ \\
		Dataset            & Alibaba Cluster Dataset \\
		Number of trials   & $10$ \\
		\bottomrule
	\end{tabular}
\end{table}

The evaluation has been conducted using Alibaba Cluster Dataset~\cite{alibaba}. We have used the tasks' batch that has up to $14295730$ tasks. The computation workload has been implemented to match the tasks in the dataset. These tasks are grouped into 12 distinct services. Table~\ref{tab:frequency_redundancy} shows the frequency and redundancy of each type of service. We have also implemented Least Frequency Used (LFU) policy in order to clean the Reuse Table by removing the entries that have a smaller frequency rate and keep room for popular inputs.

\begin{table}[!b]
	\centering
	\makegapedcells
	\caption{Service's Redundancy.}
	\label{tab:frequency_redundancy}
	\begin{tabular}{l l c}
		\toprule
		\textit{Service} &
		\textit{Frequency} &
		\textit{Redundancy (\%)} \\
		\midrule
		
		S1 & 12		& 0.00008 \\
		S2 & 488	& 0.00341 \\
		S3 & 5519	& 0.03861 \\
		S4 & 6543	& 0.04577 \\
		S5 & 8889	& 0.06218 \\
		S6 & 28777	& 0.20130 \\
		S7 & 35061	& 0.24526 \\
		S8 & 53933	& 0.37727 \\
		S9 & 322929	& 2.25892 \\
		S10 & 399489	& 2.79446 \\
		S11 & 1226388	& 8.57870 \\
		S12 & 12207702	& 85.39404 \\
		
		\bottomrule
	\end{tabular}
\end{table}

\subsection{Evaluated Metrics}
We present the \nth{90} percentile of the results collected after running 10 trials. Hereafter, we measured the following metrics:

\begin{itemize}
	\item \textit{Task completion time}: refers to the sum of communication and computation time to execute a task. It is calculated as the sum of the time elapsed between the generation of a task by the end-user, the execution of the task (either at the edge or the cloud server), and the reception of the results by the requested end-user.
	
	\item \textit{Task computation time}: refers to the time elapsed between the reception of the task by the edge or cloud server and the production of the result. This metric indicates the computation time (\eg computation from scratch or computation with reuse) and the waiting time before execution.
	
	\item \textit{Resource utilization}: refers to the claimed capacity of the resources utilized by the server to execute the received tasks.
	
	\item \textit{Computation load}: refers to the amount of computation executed at the cloud, at the edge, and at the edge after performing computation reuse to satisfy received tasks.
	
	\item \textit{Reduction rate:} refers to the amount of reduced communication and computation while using computation reuse at the edge. This metric indicates the gain earned with \sol and \esol schemes.
\end{itemize}

\subsection{Experimental Results}
Fig.~\ref{fig:completion_time} shows the average completion time when the number of tasks increased. The results indicate that completing the execution of tasks at the cloud has the largest delay since transmitting large size of inputs data requires more time especially when the core network is congested. Indeed, offloading the computation at the edge may provide better results, but only when selecting the suitable service to offload. {\boubakr For instance, random computation offloading at the edge exceeds the large time produced by the cloud. This is due to the randomness of decision-making. On the other hand, greedy, simulated annealing, genetic algorithm, and matching schemes without reuse perform better than the cloud since the offloading decisions are well-rounded, but with better strategies that address that consider the server capacity and distribution of tasks. We also notice that simulated annealing is better than greedy since it uses a probabilistic technique to approximate a global optimum. Similarly, the matching theory is better than the genetic algorithm due to the scalability of matching services with edge servers without performing thermodynamics principles. The same behavior occurs when computation reuse is applied. Indeed, applying computation reuse at the edge grants better results since more tasks gain less computation. \sol outperforms all schemes because of the optimal offloading decision that takes into account the reusability of computation. Moreover, we can notice that \sol exemplifies an inverse correlation, where the number of tasks increases the completion time decreases. This is due to the redundancy in input data that leads to more computation reuse and less computation overload.}

\begin{figure}[!t]
	\centering
	\includegraphics[width=.9\linewidth]{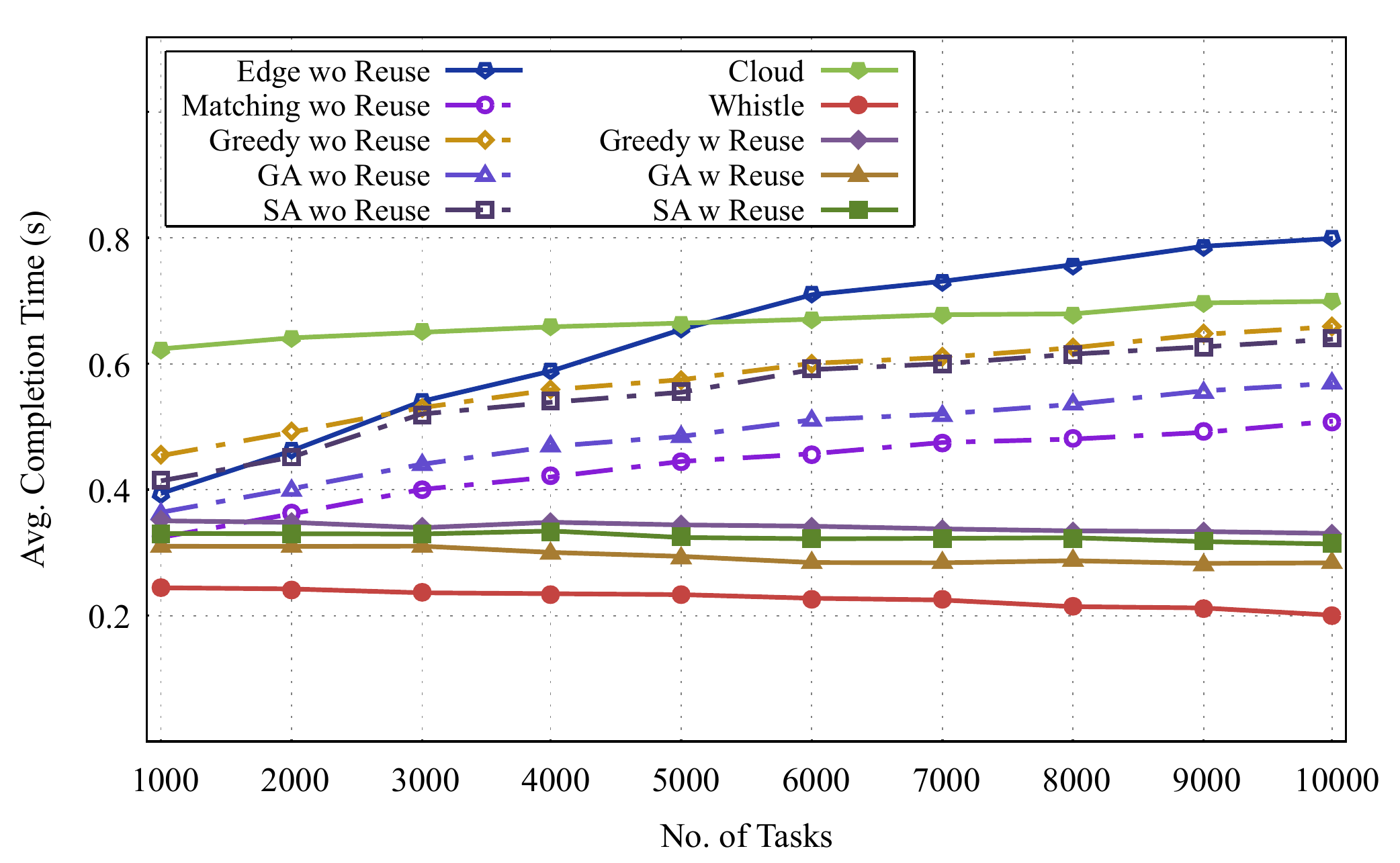}
	\caption{Tasks completion time.}
	\label{fig:completion_time}
\end{figure}

Fig.~\ref{fig:computation_time} presents the average computation time when the number of tasks increases. The results show that random offloading at the edge has the worse output since most of the offloaded services have fewer invocations and hence tasks need to be bounced back after a long waiting time for remove execution at the cloud. {\boubakr Similarly, greedy, simulated annealing, genetic algorithm, and matching schemes without reuse have better results compared with random offloading, yet the computation witnesses more waiting time to satisfy all received tasks with limited computation resources. On the other hand, cloud computing had better results when considering a small number of tasks, but when the number of tasks increases, and even though with unbounded resources, the cloud does not outperform \sol. The latter benefits from computation reuse and performs only the lookup process at Reuse Table instead of computation from scratch. A lookup operation has a very light load compared with the task's execution. We can also notice that when the number of tasks increases, the average computation time decreases due to the fact that the received tasks have more redundancy in their inputs that lead to more computation reuse.}

\begin{figure}[!t]
	\centering
	\includegraphics[width=.9\linewidth]{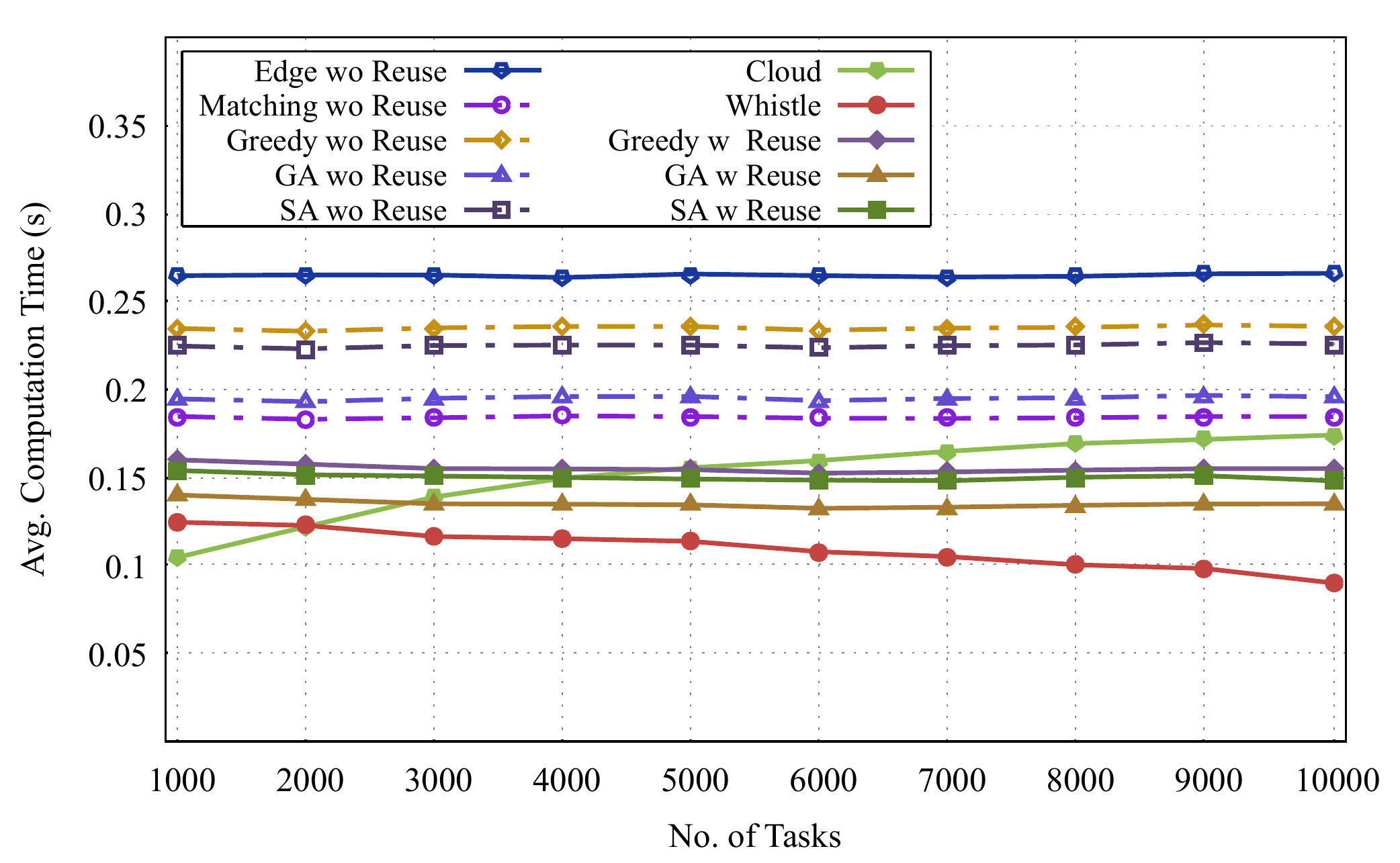}
	\caption{Task's Computation Time.}
	\label{fig:computation_time}
\end{figure}

Fig.~\ref{fig:resource_utilization} shows the percentage of resource utilization when executing tasks. Here, we fix the number of tasks to 10000 and for each trial, we increase the capacity of the edge server and observe the resource utilization. {\boubakr The results demonstrate that offloading schemes without reuse utilize the entire resources since more computation is required to satisfy the received tasks. When computation reuse is applied, we can notice that fewer resources (with variations based on the offloading scheme) are utilized since the reusability of services has not been taken into consideration. On the other hand, \sol utilizes only 16\% of resources to satisfy all received tasks. This is due to the fact that lookup operations on tables are not costly compared with a full task's execution.}

\begin{figure}[!t]
	\centering
	\includegraphics[width=.9\linewidth]{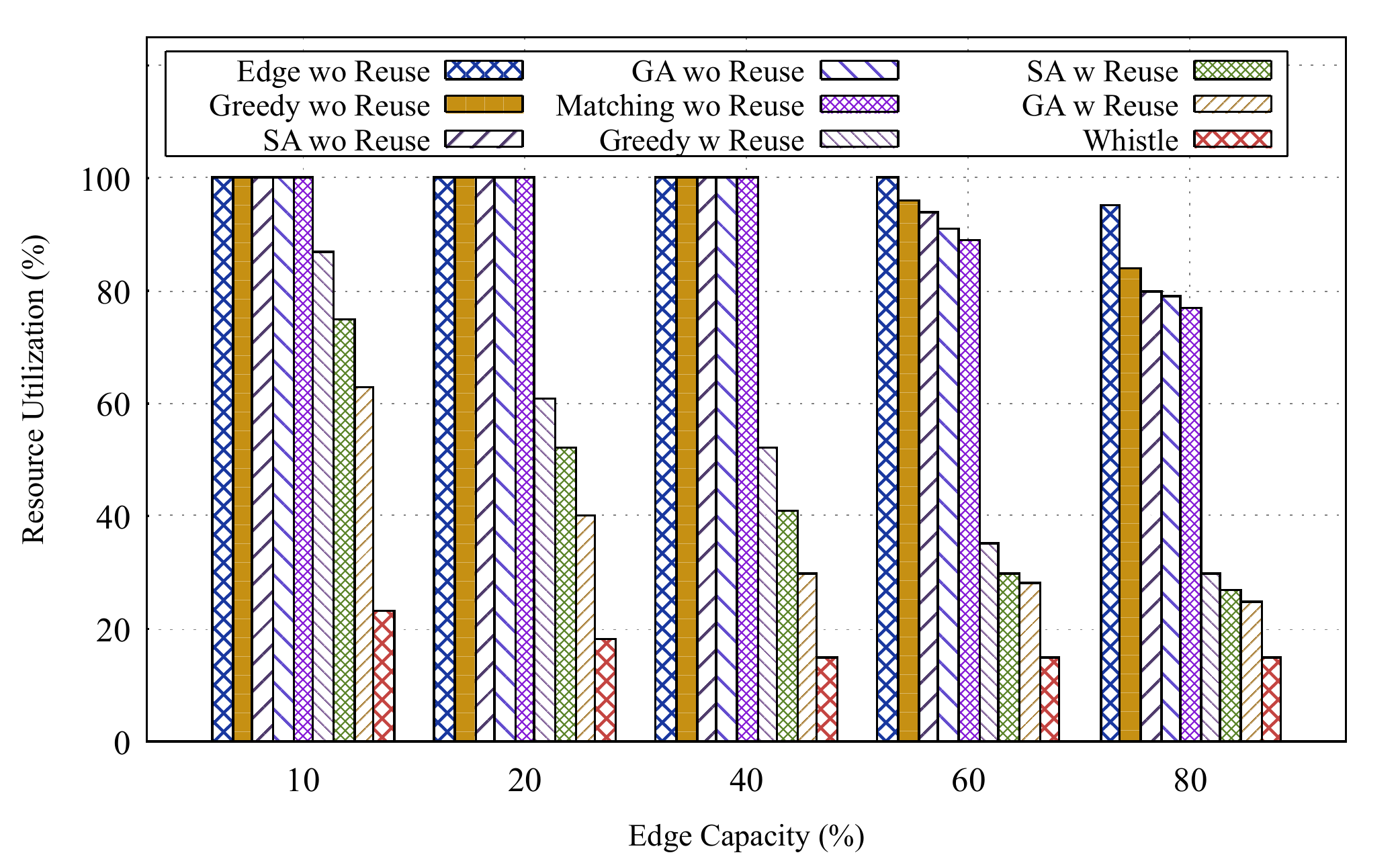}
	\caption{Resource utilization.}
	\label{fig:resource_utilization}
\end{figure}

Fig.~\ref{fig:computation_load} proves afore-obtained results and depicts the normalized load. We can notice that in the offloading schemes without computation reuse, a small ratio of computation is executed at the edge while more computation is executed at the cloud. This is due to (i) not all the appropriate services are offloaded to the edge and/or (ii) the offloaded services could not satisfy a large number of received tasks with the available resources. {\boubakr  However, when we increase the edge capacity, less computation is bounced back to the cloud since the edge server can accommodate more services. We also notice that by applying computation reuse, the results are improved compared with conventional computation. Finally, \sol efficiently utilizes the edge resources by offloading the most invoked services with high reusability at the edge.} Thus, most of the computation (almost 99\%) is executed at the edge server. Combined with Fig.~\ref{fig:resource_utilization}'s outputs, we can conclude that \sol efficiently and effectively utilizes the edge resources.

\begin{figure}[!t]
	\centering
	\includegraphics[width=.9\linewidth]{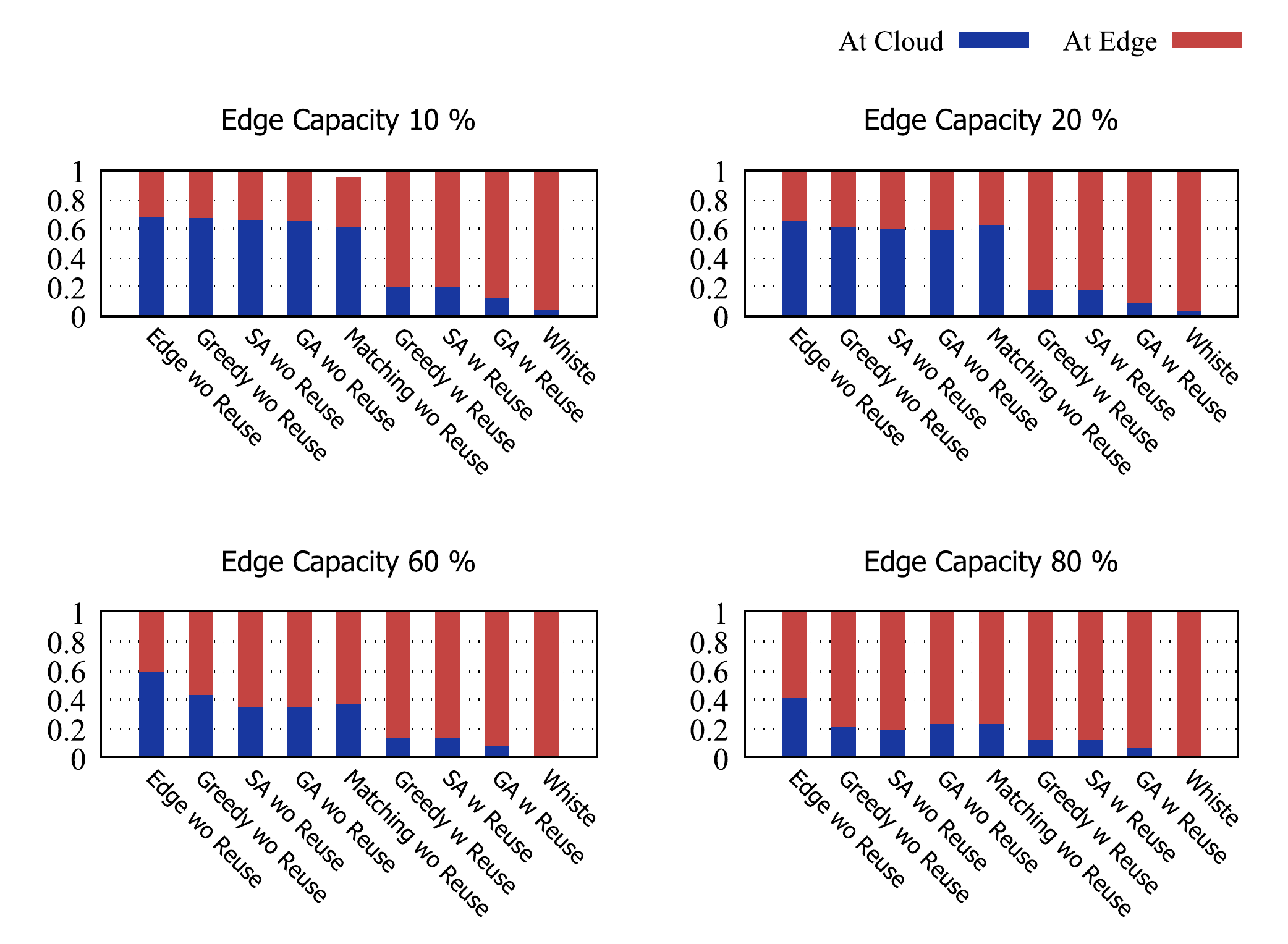}
	\caption{Computation load.}
	\label{fig:computation_load}
\end{figure}

Fig.~\ref{fig:reduction_rate} shows the reduction rate, for both communication and computation, when the number of tasks increased. The results show that with \sol, we could reduce the communication up to 71\% since it restrains the number of traversed data in the core network. Similarly, \sol reduces the computation (workload) up to 77\% since it eliminates redundant computation and uses only lightweight lookup operations instead.

\begin{figure}[!t]
	\centering
	\includegraphics[width=.9\linewidth]{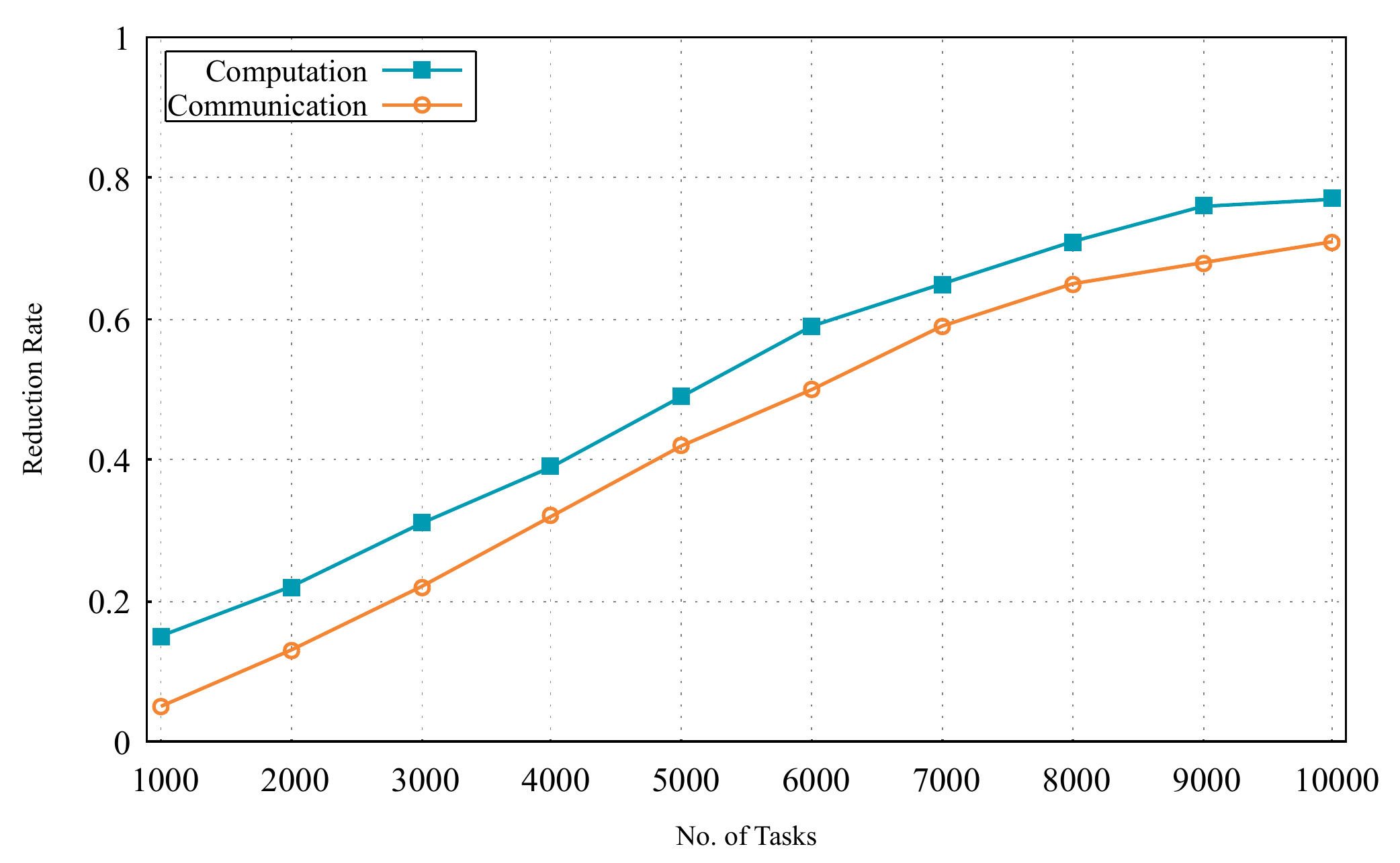}
	\caption{Reduction rate.}
	\label{fig:reduction_rate}
\end{figure}

Although \sol achieves better results compared to other schemes, some tasks are bounced back to the cloud for remote computation. Their associated services are not offloaded to the local edge server. Fig.~\ref{fig:completion_time_ext} evaluates the performance of \esol. From the results, we can see that, with \esol, and just by forwarding the computation to the near edge server, we could improve the completion time by 13\%, since bouncing the computation to the cloud consumes more times compared with forwarding it to the near edge server.

\begin{figure}[!t]
	\centering
	\includegraphics[width=.9\linewidth]{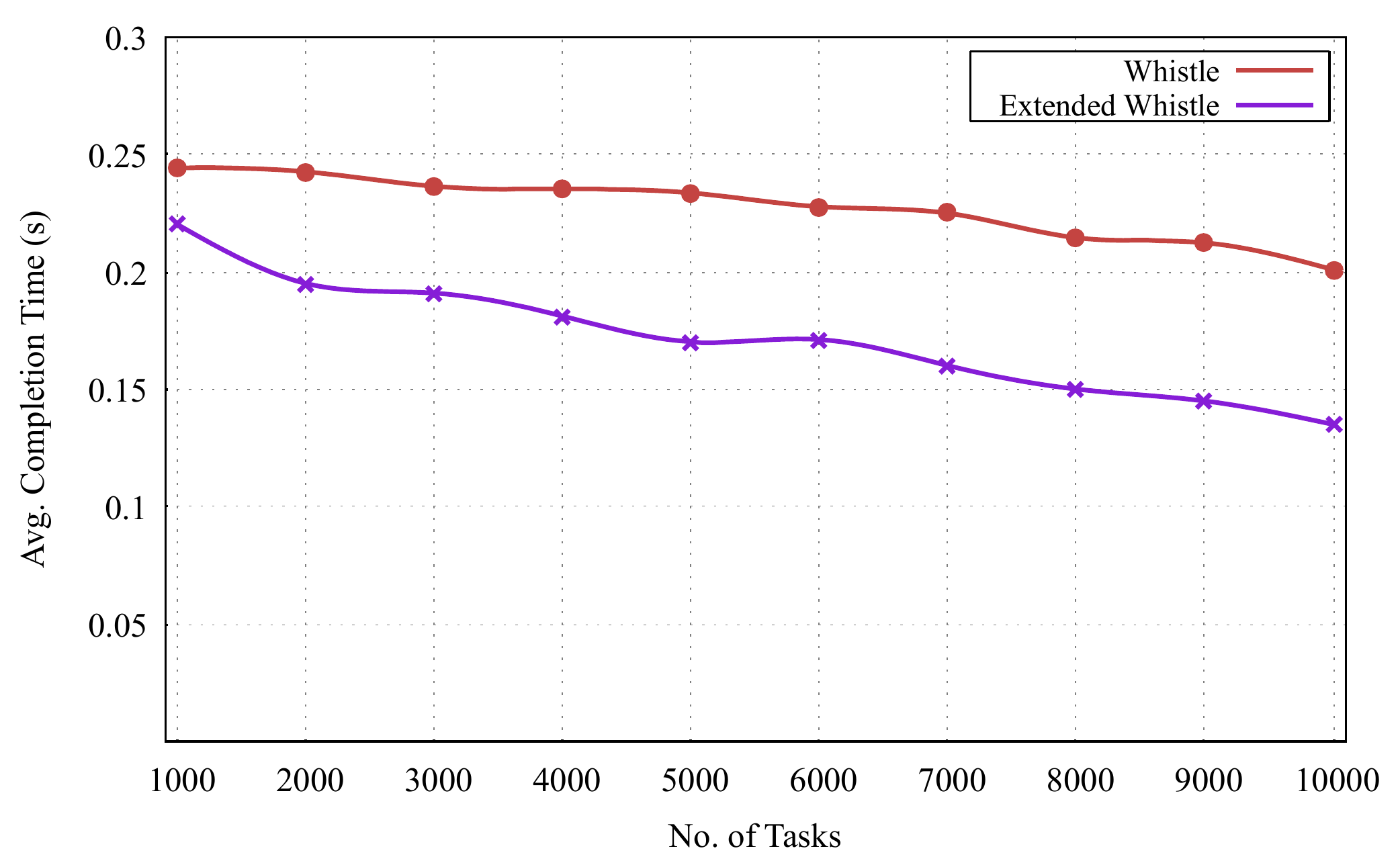}
	\caption{Task's completion time.}
	\label{fig:completion_time_ext}
\end{figure}

\section{Related Work}
\label{sec:related_work}
Edge computing aims at moving computation and data analysis from the cloud to edge devices, which will increase the system's overall analytic capacity. In fact, edge computing has shown a great increase in capacity and resiliency, gain better security and privacy protection, and leverage the 5G network's reduced latency.
Computation offloading is one of the active research in edge computing. For instance, 
Xu~\etal \cite{xu2018joint} focused on computation offloading from end-users to edge servers. The edge server caches different services including their libraries so that the offloaded computation can be executed. The authors designed an online algorithm that jointly optimizes dynamic services caching and task offloading in dense cellular networks.
Chen~\etal \cite{chen2017joint} designed a three-step algorithm to jointly manage computation and communication resources aiming to minimize the overall energy, computation, and delay. The designed algorithm comprises a semi-definite relaxation, alternating optimization, and sequential tuning.
Sundar~\etal \cite{sundar2018offloading} focused on offloading dependent tasks with communication delay and deadline constraints. The authors designed a heuristic-greedy allocation scheme that tends to allocate a completion deadline for each task and greedily optimizes the scheduling of each task subject to its time allowance.
Zhang~\etal \cite{zhang2020dynamic} investigated the task offloading and resource allocation problem in ultra-dense networks. The authors designed an online task offloading and resource allocation scheme based on Lyapunov optimization using decomposition methods, matching games, and geometric programming.

With the massive increase in connected devices and the nature of applications, various attempts have been proposed to further enhance the performance of edge computing. Computation reuse is one of the pillar mechanisms used in compute-less networking. Computation reuse aims at caching the inputs and output of popular computation and use it to fulfill the newly coming tasks that have similar input data. By doing so, a large number of redundant computations will be eliminated and hence both computation resources and end-to-end delay will be improved.
{\boubakr Al Azad~\etal~\cite{azad2021promise} discussed the promise and challenges to achieve pervasive computation deduplication and reuse at the edge of the network. The concept of computation deduplication aims at inferring whether computation reuse is possible for the arrived task compared with the already executed/stored tasks.}
Drolia~\etal \cite{drolia2017cachier} designed an image recognition system to meet the latency-sensitive nature of the application. The designed system uses a web cache to store different computation at the edge and use it to recognize similar inputs. The authors also balanced computation between edge and cloud so that they achieve larger latency.
Lee~\etal~\cite{lee2019case} presented an empirical study for using computation reuse at edge networks. The authors motivated their studies with different use cases and scenarios where users send redundant tasks (\ie redundant input data). Through this study, the authors showed computation reuse at the edge has the potential to reduce resource utilization and lessen the time needed for the completion.
Guo~\etal~\cite{guo2018potluck, guo2018foggycache} designed a compute-less architecture that eliminates redundant architecture by performing approximate deduplication. The designed architecture uses locality-sensitive hashing to identify similar input data, store the computation, and then share the processing results between applications. Although the architecture maximizes deduplication opportunities, cross-sharing the computation between applications requires further investigation in security and privacy.
Other efforts have been proposed to merge next-generation Internet architecture, such as Information-Centric Networking (ICN) or Named Data Networking at the edge networking by leveraging computation reuse. Mastorakis~\etal \cite{mastorakis2020ICedge} studied the merger of edge computing with ICN by using naming abstraction. The authors proposed to use computation reuse at the edge to satisfy tasks at the edge with the minimum cost.

Game theory is an important tool for decision-making in communication networks~\cite{moura2018game}. Various attempts have been presented in the literature that uses game theory for optimal assignment tasks. For instance,
{\boubakr Zhang~\etal~\cite{zhang2020dynamic} studied task computation offloading in mobile edge networks. The authors used matching theory for sub-channel allocation in order to optimize the efficiency of network energy and service delay.}
Amine~\etal~\cite{amine2017many} used matching theory to address the energy issues in mobile cellular ultra-dense networks. The authors modeled the co-channel deployment problem as a many-to-many matching game and then introduced a many-to-many stable matching scheme that assigns each macro-indoor user to the most suitable base station.
Filali~\etal~\cite{filali2018sdn} used a one-to-many matching game to assign each switch to SDN controllers with a minimum quota. The quota represents the utilization of the processing capacity of the controller.
Liu~\etal~\cite{liu2018computation} focused on multi-user computation offloading in vehicular networks. The problem is modeled as a game theory and a distributed algorithm is designed to provide an efficient computation offloading.
Similarly, Messous~\etal~\cite{messous2019game} studied computation offloading in unmanned aerial vehicle networks and designed a non-cooperative game that tends to optimally provide computation locally, at the edge server, or powerful edge server so that achieve better performance.

\section{Conclusion}
\label{sec:conclusion}
In this paper, we designed \sol, a compute-less driven network that adopts service offloading and computation reuses concepts. \sol is a many-to-many matching game that enables service providers to offload their most invoked and highly reusable services to the edge servers so that meet the desired quality of service requirements. We further extended \sol to provide a one-to-many matching and allow computation sharing among near edge servers, and hence eliminate bouncing the computation back to the cloud. The designed schemes have been extensively evaluated using real-world datasets. We were able to achieve the best performance in terms of task completion time, resource utilization, computation load, and computation/communication reduction against various offloading schemes due to the efficiency of matching strategy and computation caching/reuse. We also provided theoretical analyses to prove the stability of the designed strategies.

\section*{Acknowledgments}
The authors would like to thank the Natural Sciences and Engineering Research Council of Canada (NSERC) for the financial support of this research.

\bibliographystyle{IEEEtran}

\end{document}